%% file: main.tex
\documentclass[12pt,a4paper]{article}

\usepackage{amsmath,amsthm,amssymb,amscd,a4wide}

\usepackage{dsfont}
\usepackage{mathrsfs}
\usepackage{setspace}
\usepackage{hyperref}

\input{def.tex}

\numberwithin{equation}{section}

\begin{document}

\thispagestyle{empty}

\vspace*{1cm}

\begin{center}

{\LARGE\bf On Lennard-Jones-type potentials on the half-line} \\

\vspace*{2cm}

{\large Federica Gregorio \footnote{E-mail address: {\tt fgregorio@unisa.it}} }%

\vspace*{5mm}

Dep. of Information Engineering, Electrical Engineering and Applied Mathematics\\
Universit\`{a} degli Studi di Salerno \\
Fisciano (SA)\\
Italy\\

\vspace*{10mm}

{\large and }

\vspace*{10mm}

{\large Joachim Kerner \footnote{E-mail address: {\tt joachim.kerner@fernuni-hagen.de}} }%

\vspace*{5mm}

Department of Mathematics and Computer Science\\
FernUniversit\"{a}t in Hagen\\
58084 Hagen\\
Germany\\

\end{center}

\vfill

\begin{abstract} In this paper we study a particle under the influence of a Lennard-Jones potential moving in a simple quantum wire modelled by the positive half-line. Despite its physical significance, this potential is only rarely studied in the literature and due to its singularity at the origin it cannot be considered as a standard perturbation of the one-dimensional Laplacian. It is therefore our aim to provide a thorough description of the full Hamiltonian in one dimension via the construction of a suitable quadratic form. Our results include a discussion of spectral and scattering properties which finally allows us to generalise some results from~\cite{RobinsonSingularI} as well as~\cite{RadinSimon}.
\end{abstract}

\newpage

\input{intro}

\input{1sec}

\vspace*{0.5cm}

\subsection*{Acknowledgement}{The authors are very happy to thank R.~Weder (Universidad Nacional Aut\'{o}noma de M\'{e}xico) for many helpful comments on the manuscript. JK also wants to thank S.~Egger (Technion, Israel) for helpful discussions. We also want to thank the referee for pointing out interesting references.} 

\vspace*{0.5cm}

{\small
\bibliographystyle{amsalpha}
\bibliography{Literature}}

\end{document}

%% file: def.tex
\newcommand{\ud}{\mathrm{d}}

\newcommand{\cD}{{\mathcal D}}

\numberwithin{equation}{section}

\newtheorem{theorem}{Theorem}[section]
\newtheorem{lemma}[theorem]{Lemma}
\newtheorem{prop}[theorem]{Proposition}

\newtheorem{remark}[theorem]{Remark}

\theoremstyle{definition}


%% file: intro.tex
\section{Introduction}
In this note we are concerned with the Schr\"{o}dinger operator 
\begin{equation}\label{SchrodingerOperator}
-\frac{\ud^2}{\ud x^2}+\left(\frac{\alpha}{|x|^{12}}-\frac{\beta}{|x|^{6}}\right)\ ,
\end{equation}
defined on the Hilbert space $L^2(\mathbb{R}_+)$ with $\alpha,\beta > 0$ some constants. The potential term in~\eqref{SchrodingerOperator} is the Lennard-Jones potential, arguably one of the most important potentials in solid-state physics which is used to describe the interaction between two neutral atoms \cite{LennardJones,DobbsJones}. Its main application lies with the description of the crystallisation of noble gases such as, e.g., argon. Note that the second term in the potential is the attractive van der Waals interaction (dipole-dipole interaction) and the first term is due to the repulsion at small distances which itself is a consequence of the Pauli exclusion principle. 

Of course, using a separation of variables into relative and center-of-mass coordinates, it is natural to consider the operator~\eqref{SchrodingerOperator} given one wanted to describe two neutral atoms on the full line $\mathbb{R}$. However, another interpretation is obtained by regarding $\mathbb{R}_+=[0,\infty)$ as a quantum wire (or quantum graph) with a vertex at the origin $x=0$ at which one imagines an additional complex internal structure such as a quantum dot~\cite{HarrisonValavanis}. Here a quantum dot has to be thought of as a relatively small box with hard walls at which another particle is placed (``particle in the box''). If one then assumes to first approximation that the particle in the box is not excited through the interaction with the other particle moving in the wire, the potential experienced by the particle in the wire is exactly of the Lennard-Jones type.

Also, from a mathematical point of view the Lennard-Jones potential is interesting due to its high degree of singularity at the origin which implies that it is not relatively bound with respect to the Laplacian, i.e., it is not of Kato class. This also implies that known results in the scattering theory for integrable and relatively bound potentials do not apply which then motivated Robinson to study highly singular potentials in more detail \cite{RobinsonSingularI}. In particular, he proved the existence and completeness of the wave operators for the Lennard-Jones potential in three dimensions [Theorem~5.6,\cite{RobinsonSingularI}]. It is one of the aims of this paper to generalise this result to one dimension and even generalising it by proving asymptotic completeness~\cite{BEH08}, see Section~\ref{SectionScattering}.

In addition, motivated by the classical $N$-body problem, Radin and Simon characterised, for $H=-\Delta+V$ on $L^2(\mathbb{R}^m)$ with $V$ in the Kato class, subspaces which are invariant under the time-evolution operator $e^{-iHt}$ which, in particular, implies explicit upper bounds on 
\begin{equation*}
\|x e^{-iHt}\varphi\|_{L^2(\mathbb{R}^m)}\ , \quad \varphi \in L^2(\mathbb{R}^m)\ ,
\end{equation*}
i.e., the expectation value of the mean distance from the origin showing that this norm remains finite in finite time. In Section~\ref{SectionInvariant} we will generalise this result to our setting, i.e., for a potential which is not of Kato class.

In Section~\ref{SectionFormulationModel} we start by establishing a rigorous realisation of the operator~\eqref{SchrodingerOperator} via the construction of a suitable quadratic form. We characterise the form as well as the operator domain explicitly, also proving $H^2$-regularity and essential self-adjointness of the minimal operator. In Section~\ref{SectionSpectrum} we then study spectral properties, characterising the discrete as well as the essential part of the spectrum. In addition, we prove that the essential part is purely absolutely continuous. 

We note that the treatment of singular potentials with quadratic form methods goes back to Simon's thesis (1971); see also \cite{NeidhardtZagrebnov} and references therein.

%% file: 1sec.tex
\section{Formulation of the model}\label{SectionFormulationModel}
We consider a (spinless) particle moving on the half-line $\mathbb{R}_+=[0,\infty)$ under the influence of an external potential of the Lennard-Jones type. More explicitly, the Hamiltonian of the particle shall be given by
\begin{equation}\label{Hamiltonian}
H_{\alpha,\beta}:=-\frac{\ud^2}{\ud x^2}+V(x):=-\frac{\ud^2}{\ud x^2}+\frac{\alpha}{|x|^{12}}-\frac{\beta}{|x|^{6}}\ ,
\end{equation}
with $\alpha,\beta > 0$. The quadratic form formally associated with this Hamiltonian is given by
\begin{equation*}\label{Form}
q_{\alpha,\beta}[\varphi]:=\int_{\mathbb{R}_+}|\nabla \varphi|^2 \ \ud x+\alpha\int_{\mathbb{R}_+}\frac{|\varphi|^2}{|x|^{12}}\, \ud x-\beta\int_{\mathbb{R}_+}\frac{|\varphi|^2}{|x|^{6}}\, \ud x\ .
\end{equation*}
\begin{theorem}\label{TheoremForm} The form $q_{\alpha,\beta}$ defined on
	\begin{equation}\label{FormDomain}
	\cD_{q}=\left\{\varphi \in H^1(\mathbb{R}_+): \int_{\mathbb{R}_+}V(x)|\varphi(x)|^2 \ \ud x < \infty  \right\}
	\end{equation}
 is densely defined, closed and bounded from below on $L^2(\mathbb{R}_+)$.
\end{theorem}
\begin{proof} Since $C^{\infty}_0(\mathbb{R}_+)$ is dense in $L^2(\mathbb{R}_+)$, density follows readily. 
	
	In a next step we realise that $V(x)=\frac{\alpha}{|x|^{12}}-\frac{\beta}{|x|^{6}} \geq 0$ for all $x \leq x_0:=\sqrt[6]{\frac{\alpha}{\beta}}$. Furthermore, $V(x) \geq V\left(\sqrt[6]{\frac{2\alpha}{\beta}}\right)$ for all $x\in\mathbb{R}_+$ which implies, setting $\gamma:=V\left(\sqrt[6]{\frac{2\alpha}{\beta}}\right)$,
	\begin{equation*}
	q_{\alpha,\beta}[\varphi] \geq -\left|\gamma\right|\cdot \|\varphi\|^2_{L^2(\mathbb{R}_+)}
	\end{equation*}
	and hence the form is bounded from below. 
	
	Now, let $(\varphi_n)_{n \in \mathbb{N}} \subset H^1(\mathbb{R}_+)$ be a Cauchy sequence with respect to the form norm $\|\cdot\|^2_{q}:=q_{\alpha,\beta}[\cdot]+\left[\left|\gamma\right|+1\right]\cdot \|\cdot\|^2_{L^2(\mathbb{R}_+)}$. Due to completeness of $H^1(\mathbb{R}_+)$ there exists a function $\varphi \in H^1(\mathbb{R}_+)$ such that $\varphi_n \rightarrow \varphi$ in $H^1(\mathbb{R}_+)$-norm. Furthermore, employing the Lemma of Fatou we obtain
	\begin{equation}\label{EquationI}\begin{split}
	\left|\int_{\mathbb{R}_+}V(x)|\varphi(x)-\varphi_n(x)|^2\ \ud x\right|&=\int_{(0,x_0)}V(x)|\varphi(x)-\varphi_n(x)|^2\ \ud x\ \\ 
	& \quad +\int_{(x_0,\infty)}\left|V(x)\right| |\varphi(x)-\varphi_n(x)|^2\ \ud x \\
	\ \ &\leq \liminf_{k \rightarrow \infty}\int_{(0,x_0)}V(x)|\varphi_{n_k}(x)-\varphi_n(x)|^2\ \ud x \\
	& \quad + \|V|_{(x_0,\infty)}\| \cdot \|\varphi-\varphi_n\|^2_{L^2(x_0,\infty)}\\
	&\leq \varepsilon\ ,
	\end{split}
	\end{equation}
	for $k,n \geq n_0$, $n_0 \in \mathbb{N}$, since $(\varphi_n)_{n \in \mathbb{N}}$ is Cauchy with respect to the form norm. Using~\eqref{EquationI} we readily conclude that $\varphi \in \cD_q$ and $q_{\alpha,\beta}[\varphi_n-\varphi] \rightarrow 0$ as $n \rightarrow \infty$.\\
\end{proof}
Since the form domain~\eqref{FormDomain} is not very explicit, we aim to characterise it further. In a first result we show that all functions in this domain fulfil Dirichlet boundary conditions at zero.
\begin{prop}\label{PropZero} For $\varphi \in \cD_q$ one has $\varphi(0)=0$. 
\end{prop}
\begin{proof} We first note that $|\varphi(0)| < \infty $ due to the trace theorem~\cite{Dob05} for Sobolev functions.
	
	 Now assume that $|\varphi(0)| > 0$: Since $H^1(\mathbb{R}_+)$-functions are continuous, we conclude the existence of a small interval $[0,\delta)$, $\delta > 0$, such that $|\varphi(x)| > \varepsilon$ for all $x \in [0,\delta)$ and some $\varepsilon > 0$. This however immediately implies that 
	 \begin{equation*}
	 \int_{0}^{\delta}V(x)|\varphi(x)|^2 \ \ud x = \infty
	 \end{equation*}
	 and hence $\varphi \notin \cD_q$.
\end{proof}
%
%
By the representation theorem of quadratic forms \cite{BEH08} there exists a unique self-adjoint operator $A$ with domain $\cD(A) \subset \cD_q$ being associated with $q_{\alpha,\beta}$. As shown in the next statement, this operator is indeed given by~\eqref{Hamiltonian}.
\begin{lemma}\label{DomainOperator} The operator $A$ associated with $q_{\alpha,\beta}$ coincides with the operator $H_{\alpha,\beta}$.
%
\end{lemma}
\begin{proof} According to the representation theorem of forms~\cite{BEH08} we have the abstract characterisation
\begin{equation}\label{AbstractCharacterisationOperator}\begin{split}
\cD(A)&:=\{\varphi\in \cD_q:\,\exists \psi\in L^2(\mathbb{R}_+)\ \textrm{s.t.}\ q_{\alpha,\beta}[\varphi,h]=(\psi,h)\ \textrm{for all}\ h\in\cD_q\}\ ,\\
A\varphi&:=\psi\ .
\end{split}
\end{equation}
Now, let $\varphi\in \cD(A)$ be given: Then $\varphi_n:=\chi_{(\frac{1}{n},\infty)}\varphi \in H^2(\frac{1}{n},\infty)$ (note that this follows, e.g., using the difference quotient technique~\cite{GilTru83,Dob05} which is the standard technique to show local $H^2$-regularity
) and an integration by parts yields
\begin{equation*}\begin{split}
q_{\alpha,\beta}[\varphi_n,h]&=\int_{\frac{1}{n}}^\infty \overline{\nabla \varphi_n}\cdot \nabla h\, \ud x+\alpha\int_{\frac{1}{n}}^\infty\frac{\overline{\varphi_n} h}{|x|^{12}}\, \ud x-\beta\int_{\frac{1}{n}}^\infty\frac{\overline{\varphi_n} h}{|x|^{6}}\, \ud x\\
&=-\int_{\frac{1}{n}}^\infty \overline{\Delta \varphi_n}  h\, \ud x+\alpha\int_{\frac{1}{n}}^\infty\frac{\overline{\varphi_n} h}{|x|^{12}}\, \ud x-\beta\int_{\frac{1}{n}}^\infty\frac{\overline{\varphi_n} h}{|x|^{6}}\, \ud x\\
&=(H_{\alpha,\beta} \varphi_n, h)\ .
\end{split}
\end{equation*}
By~\eqref{AbstractCharacterisationOperator} one then obtains
\[q_{\alpha,\beta}[\varphi_n,h]=(H_{\alpha,\beta}\varphi_n,h)=(A\varphi_n,h)\]
and therefore
\begin{equation*}
(A\varphi)|_{(\frac{1}{n},\infty)}=(H_{\alpha,\beta}\varphi)|_{(\frac{1}{n},\infty)} 
\end{equation*}
by choosing $h \in C^{\infty}_0(\frac{1}{n},\infty)$ and taking density of $C^{\infty}_0(\mathbb{R}_+)$ in $L^2(\mathbb{R}_+)$ into account.
\end{proof}
\begin{remark} Based on Lemma~\ref{DomainOperator} we set $\cD(H_{\alpha,\beta}):=\cD(A)$ and write $A=H_{\alpha,\beta}$.
\end{remark}
Until now we have identified one self-adjoint realisation of $(H_{\alpha,\beta},C^{\infty}_0(\mathbb{R}_+))$ through the construction of a suitable quadratic form. The following result shows that this is indeed the only one existing.
\begin{prop}\label{EssentialSelfAdjointness}
The operator $(H_{\alpha,\beta},\cD(H_{\alpha,\beta}))$ is the unique self-adjoint extension of $(H_{\alpha,\beta},C^{\infty}_0(\mathbb{R}_+))$.
\end{prop}
\begin{proof} This follows from the theory of Sturm-Liouville operators as presented in~\cite{schmudgen2012unbounded}. In particular, [Propositions~15.11,~15.12,\cite{schmudgen2012unbounded}] show that $(H_{\alpha,\beta},C^{\infty}_0(\mathbb{R}_+))$ is in the so-called limit point case at zero and at infinity. The statement then follows with [Theorem~15.10,\cite{schmudgen2012unbounded}] which shows that $(H_{\alpha,\beta},C^{\infty}_0(\mathbb{R}_+))$ has deficiency indices $(0,0)$.
\end{proof}
In a next step we characterise the domain of $H_{\alpha,\beta}$ in more detail by proving that $D(H_{\alpha,\beta}) \subset H^2(\mathbb{R}_+)$. In other words, we establish (global) $H^2$-regularity~\cite{Gri85,GilTru83,Dob05}. We also prove that the derivative at zero vanishes, i.e., functions in the operator domain fulfil also Neumann boundary conditions.
\begin{theorem}\label{PropZeroI} One has $D(H_{\alpha,\beta}) \subset H^2(\mathbb{R}_+)$. Furthermore, if
	$\varphi \in D(H_{\alpha,\beta})$ then $\varphi'(0)=0$. 
\end{theorem}
\begin{proof} We first prove that $D(H_{\alpha,\beta}) \subset H^2(\mathbb{R}_+)$: We first note that every $\varphi \in \cD(H_{\alpha,\beta})$ is locally in $H^2(\mathbb{R}_+)$ (see the proof of Lemma~\ref{DomainOperator}) and hence $\varphi^{\prime \prime}$ exists as a weak derivative.  Proposition~\ref{EssentialSelfAdjointness} then shows that $(H_{\alpha,\beta},\cD(H_{\alpha,\beta}))=\overline{(H_{\alpha,\beta},C^{\infty}_0(\mathbb{R}_+))}$ with respect to the operator norm. Hence, for any $\varphi \in \cD(H_{\alpha,\beta})$ there exists a sequence $(\varphi_n)_{n \in \mathbb{N}} \subset C^{\infty}_0(\mathbb{R}_+)$ such that 
	\begin{equation*}
	\|-\varphi^{\prime \prime}+V\varphi+\varphi_{n}^{\prime \prime}-V\varphi_n\|_{L^2(\mathbb{R}_+)} \leq \varepsilon
	\end{equation*}
	for any $\varepsilon > 0$ and $n$ large enough. Hence, by triangle inequality 
	\begin{equation*}\begin{split}
	\|\varphi^{\prime \prime}\|_{L^2(\mathbb{R}_+)} \leq \varepsilon + \|V\varphi\|_{L^2(\mathbb{R}_+)}+\|\varphi_{n}^{\prime \prime}-V\varphi_n\|_{L^2(\mathbb{R}_+)} < \infty
	\end{split}
	\end{equation*}
	and hence $\varphi^{\prime \prime} \in L^2(\mathbb{R}_+)$. Note that $V\varphi \in L^2(\mathbb{R}_+)$ can be shown using the methods of the proof of [Proposition~3.2,\cite{MetafuneSchnaubelt}]: for this, one considers the operator $\tilde{H}_{\alpha,\beta}$ with shifted potential $\tilde{V}$ (note that this operator is self-adjoint on the same domain as $H_{\alpha,\beta}$) such that $\tilde{V}(x) > c_1$ for some $c_1 > 0$ and $|\tilde{V}^{\prime}(x)| \leq \gamma \tilde{V}^{3/2}(x)+c_{\gamma}$ for an arbitrarily small constant $\gamma > 0$ and some constant $c_{\gamma} > 0$. 
	
	We now turn to the second part of the statement: We first observe that $\varphi\in D(H_{\alpha,\beta})$ implies that $\varphi \in C^1[0,\infty)$ due to standard Sobolev embeddings. Furthermore, one has $\lim_{x \rightarrow 0}\frac{\varphi(x)}{x}=\varphi^{\prime}(0)$. 
	
	Now, assume that $|\varphi^{\prime}(0)|:=2\varepsilon > 0$. Consequently, there exists a $\delta > 0$ such that $\left|\frac{\varphi(x)}{x}\right| > \varepsilon$ as well as $V(x) > 0$ for all $x \in [0,\delta)$ which implies
	\begin{equation*}\begin{split}
	\int_{0}^{\delta}V(x)|\varphi(x)|^2\ \ud x&=\int_{0}^{\delta}x^2V(x)\left|\frac{\varphi(x)}{x}\right|^2\ \ud x \\
	& > \varepsilon^2\int_{0}^{\delta}x^2V(x)\ \ud x = \infty\ .
	\end{split}
	\end{equation*}
	This is a contradiction to $\varphi \in \cD_q$ and hence proves the statement.
\end{proof}
\section{On the spectrum of $H_{\alpha,\beta}$}\label{SectionSpectrum}
In this section we characterise the spectrum of $H_{\alpha,\beta}$ and in a first step we look at the essential part of the spectrum and prove that $\sigma_{ess}(H_{\alpha,\beta})=[0,\infty)$. Although this result is rather standard in the theory of Schr\"{o}dinger operators~\cite{SimonS}, we shall add a proof for the sake of completeness.  
\begin{theorem}[Essential spectrum]\label{TheoremEssential} We have
	\begin{equation*}
	\sigma_{ess}(H_{\alpha,\beta})=[0,\infty)\ .
	\end{equation*}
\end{theorem}
\begin{proof} 
	
	We first show that $[0,\infty) \subset \sigma_{ess}(H_{\alpha,\beta})$: Let an arbitrary $\lambda \in [0,\infty)$ be given. We use the Weyl characterisation in the version of quadratic forms~\cite{stollmann2001caught} and in this context a suitable Weyl sequence $(\varphi_n)_{n \in \mathbb{N}}$ is obtained by choosing $\varphi_n$ to be the (normalised) ground state eigenfunction to the Dirichlet Laplacian on the interval $I_n:=[a_n,a_n+L_n]$, i.e.,
	\begin{equation*}
	\varphi_n(x)=\sqrt{\frac{2}{L_n}}\sin\left(\frac{2\pi}{L_n}(x-a_n)\right)\ , \quad x\in I_n\ ,
	\end{equation*}
	and $\varphi_n(x)=0$ for $x\in \mathbb{R}_+\setminus I_n$. Given that $a_n \rightarrow \infty$ as $n \rightarrow \infty$ we see that $(\varphi_n)_{n \in \mathbb{N}}$ converges weakly to zero in $L^2(\mathbb{R}_+)$. Furthermore 
	\begin{equation*}
	q_{\alpha,\beta}[\varphi_n]=\frac{\pi^2}{L^2_n}+\int_{\mathbb{R}_+}V(x)|\varphi_n(x)|^2 \ud x
	\end{equation*}
	and hence $\lim_{n \rightarrow \infty}q_{\alpha,\beta}[\varphi_n]=\lambda$ if $\lim_{n \rightarrow \infty}\frac{\pi^2}{L^2_n}=\lambda$ since 
	\begin{equation*}\begin{split}
	\left|\int_{\mathbb{R}_+}V(x)|\varphi_n(x)|^2 \ud x \right|&=\left|\int_{(a_n,\infty)}V(x)|\varphi_n(x)|^2 \ud x \right| \\
	&\leq \|V\chi_{(a_0,\infty)}\| \\
	&\leq \varepsilon\ ,
	\end{split}
	\end{equation*}
	for $n$ large enough. 
	
	To prove that no value $\lambda < 0$ is contained in the essential spectrum we use an operator bracketing argument~\cite{BEH08}. More explicitly, we construct the comparison operator
	\begin{equation*}
	\tilde{H}_n:=-\frac{\ud^2}{\ud x^2}\bigg \vert_{(0,n)} \oplus \left(-\frac{\ud^2}{\ud x^2}+V(x)\right)\bigg \vert_{(n,\infty)}
	\end{equation*}
	defined on $\{\varphi \in H^2(0,n): \varphi^{\prime}(0)=\varphi^{\prime}(n)=0\}\oplus \{\varphi \in H^2(n,\infty): \varphi^{\prime}(n)=0\}$. In the sense of operators, this operator is smaller than $H_{\alpha,\beta}$ which implies that 
	\begin{equation*}
	\inf \sigma_{ess}(\tilde{H}_n) \leq \inf \sigma_{ess}(H_{\alpha,\beta})\ .
	\end{equation*}
	Since $-\frac{\ud^2}{\ud x^2}\big \vert_{(0,n)}$ has discrete spectrum only we conclude
	\begin{equation*}
	\inf \sigma_{ess}(\tilde{H}_n)=\inf \sigma_{ess}\left(\left(-\frac{\ud^2}{\ud x^2}+V(x)\right)\bigg \vert_{(n,\infty)}\right)\ .
	\end{equation*}
	Finally, since
		\begin{equation*}
		\inf \sigma_{ess}\left(\left(-\frac{\ud^2}{\ud x^2}+V(x)\right)\bigg \vert_{(n,\infty)}\right) \geq -\|V\chi_{(n,\infty)}\|_{\infty} \rightarrow 0 \ ,
		\end{equation*}
		as $n \rightarrow \infty$, we conclude the statement since $n$ can be chosen arbitrarily.
\end{proof}
We now turn attention towards the discrete part of the spectrum. 
\begin{theorem}[Discrete spectrum]\label{TheoremDiscreteSpectrum} The number of negative eigenvalues is finite. 
\end{theorem}
\begin{proof} The statement readily follows by [Theorem~5.1,\cite{berezin1991schrodinger}] taking Proposition~\ref{PropZero} into account. 
	
	%
%
%
%
%
%
\end{proof}
In a next result we show that the discrete spectrum may indeed be empty for some choices of $\alpha,\beta > 0$ even though the potential has a negative part.
\begin{theorem}[Absence of discrete spectrum] Whenever $\beta^{10/6} < 4 \alpha^{4/6}$ then
	\begin{equation*}
	\sigma_{d}(H_{\alpha,\beta})=\emptyset\ .
	\end{equation*}
\end{theorem}
\begin{proof} Proposition~\ref{PropZero} allows us to apply [Theorem~5.1,\cite{berezin1991schrodinger}] which states that the number of eigenvalues is bounded from above by the integral $\int_{0}^{\infty}x|V_{-}(x)|\ud x$. Evaluating this integral then yields the statement.
\end{proof}
In a final result we show that the essential spectrum is indeed purely absolutely continuous. 
\begin{theorem}[Absolutely continuous spectrum]\label{AbsenceSingularContinuous} We have 
\begin{equation*}
\sigma_{ac}(H_{\alpha,\beta})=[0,\infty)
\end{equation*}
and, furthermore, the spectrum is purely absolutely continuous on $[0,\infty)$.
\end{theorem}
\begin{proof} The fact that $(0,\infty)$ is purely absolutely continuous readily follows from [Theorem~1.3,\cite{Remling}] taking Proposition~\ref{PropZero} into account.
	
	Furthermore, since the singular continuous spectrum cannot be supported on the single point $\{0\}$, the statement will follow if we prove that zero is not an eigenvalue: hence assume that $\varphi \in \cD(H_{\alpha,\beta})$ is a normalised (real-valued) eigenvector to the eigenvalue zero. Since $\varphi(0)=\varphi^{\prime}(0)=0$ there exists a number $\tilde{x} > 0$ such that $\varphi^{\prime}(\tilde{x})=0$ since there exists a local maximum. Reflecting $\varphi$ across the point $\tilde{x}$ then yields an eigenfunction to eigenvalue zero of the self-adjoint operator $-\Delta+\tilde{V}$ where 
	\begin{equation*}
	\tilde{V}(x):=\begin{cases}
	V(x+\tilde{x})\ ,\quad x\geq 0 \ , \\
	V(-x+\tilde{x})\ , \quad x < 0\ .
	\end{cases}
	\end{equation*}
	However, this operator is known to have no eigenvalue zero~\cite{RammZeroEigenvalue}.
\end{proof}
\section{On the scattering theory: Existence and completeness of the wave operators}\label{SectionScattering}
In this section we discuss the scattering properties of the pair of self-adjoint Hamiltonians $(H_{\alpha,\beta},H_0)$, $H_0$ being the self-adjoint one-dimensional Laplacian on $L^2(\mathbb{R}_+)$ with Dirichlet boundary conditions at zero. In particular, we want to establish the existence and completeness of the corresponding wave operators $\Omega_{\pm}(H_{\alpha,\beta},H_0)$ and $\Omega_{\pm}(H_0,H_{\alpha,\beta})$, see~\cite{Kat66,RobinsonSingularI,BEH08} for more details. As a matter of fact, since there doesn't exist singular continuous spectrum due to Theorem~\ref{AbsenceSingularContinuous}, we will actually prove that the wave operators are asymptotically complete~\cite{BEH08}. 

 Note that the existence and completeness of the wave operators for the Lennard-Jones potential in three dimensions has been established in [Theorem~5.6,\cite{RobinsonSingularI}]. Hence the following result generalises this statement to the one-dimensional setting.
 \begin{theorem} The wave operators $\Omega_{\pm}(H_{\alpha,\beta},H_0)$ and $\Omega_{\pm}(H_0,H_{\alpha,\beta})$ exist and are complete. 
 \end{theorem}
 \begin{proof} We will prove the statement using the Birman-Kuroda Theorem~\cite{BEH08}. Also, we restrict ourselves to $\Omega_{\pm}(H_{\alpha,\beta},H_0)$, the other case being analogous. 
 	
 	We introduce some comparison operators: Let $H^{D}_{\alpha,\beta}$ denote the self-adjoint realisation of~\eqref{Hamiltonian} on $L^{2}_{\oplus}:=L^2(0,x_0) \oplus L^2(x_0,\infty)$, $x_0$ as in the proof of Theorem~\ref{TheoremForm}, with Dirichlet boundary conditions at $x_0$. Furthermore, we introduce $H_1(V)$ as the self-adjoint realisation of~\eqref{Hamiltonian} on $L^2(0,x_0)$ with Dirichlet boundary conditions at $x_0$ and $H_2(V)$ as the self-adjoint realisation of~\eqref{Hamiltonian} on $L^2(x_0,\infty)$ again with Dirichlet boundary condition at $x_0$.
 	
 	Now, taking into account Krein's resolvent formula (see, e.g., [Theorem~14.18,\cite{schmudgen2012unbounded}]) we directly conclude that $(H_{\alpha,\beta}-z)^{-1}-(H^{D}_{\alpha,\beta}-z)^{-1}$ is of finite rank and hence of trace class for $z \in \rho(H_{\alpha,\beta}) \cap \rho(H^{D}_{\alpha,\beta})$. Furthermore, on $L^{2}_{\oplus}$,
 	\begin{equation*}\begin{split}
 	\Omega_{\pm}(H^{D}_{\alpha,\beta},H^{D}_{\alpha=0,\beta=0})&=\Omega_{\pm}(H_1(V) \oplus H_2(V),H_1(0) \oplus H_2(0))\\
 	&=0 \oplus \Omega_{\pm}(H_2(V),H_2(0))\ .
 	\end{split}
 	\end{equation*}
 	In the last step we used that $H_1(V)$ has purely discrete spectrum. Now, due to the integrability of the potential in $H_2(V)$, standard results imply that the wave operators $\Omega_{\pm}(H_2(V),H_2(0))$ and hence $\Omega_{\pm}(H^{D}_{\alpha,\beta},H^{D}_{\alpha=0,\beta=0})$ are complete~\cite{Yafaev:1992,Yafaev:2010}. 
 	
 	Finally, Krein's resolvent formula implies that the wave operators $\Omega_{\pm}(H_1(0) \oplus H_2(0),H_0)$ are complete since the difference of the resolvents is of finite rank and hence of trace class. The statement then follows by the well-known chain rule for wave operators [Proposition~15.2.2,\cite{BEH08}].

 	%
 	%
 \end{proof}
 \section{On an invariant domain of the time-evolution operator and an estimate}\label{SectionInvariant}

 As in~\cite{RadinSimon} we are interested in establishing an estimate on, 
\begin{equation*}
\langle \varphi_t, x^2 \varphi_t\rangle_{L^2(\mathbb{R}_+)}\ ,
\end{equation*}
i.e., the expectation value of the square of the position operator a time $t > 0$ for arbitrary initial datum $\varphi\in L^2(\mathbb{R}_+)$. We again stress that the results of \cite{RadinSimon} are not directly applicable since the potential $V$ is not contained in the Kato class. 

Note that we write $\hat{f}\in L^2(\mathbb{R}_+)$ for the restriction of Fourier transform of
\begin{equation*}
\tilde{f}(x):=\begin{cases} f(x)\ , \quad x\geq 0\ ,\\
f(-x)\ , \quad x < 0\ ,
\end{cases}
\end{equation*}
onto $\mathbb{R}_+$ where $f \in L^2(\mathbb{R}_+)$. As a first result we establish the following, where we write $-\Delta=-\frac{\ud^2}{\ud x^2}$.
\begin{lemma}\label{TimeEstimateI} For any $\varphi \in \cD_q$ there exist constants $c,d > 0$ such that
	\begin{equation*}
	|\langle \varphi_t, x^2 \varphi_t\rangle_{L^2(\mathbb{R}_+)}| \leq (c+d\cdot t)^2 \cdot \|\varphi\|^2_{1}
	\end{equation*}
	where 
	\begin{equation*}
	\|\varphi\|^2_{1}:=\|\varphi\|^2_{L^2(\mathbb{R}_+)}+\|x\varphi\|^2_{L^2(\mathbb{R}_+)}+\|(-\Delta)^{1/2}\varphi\|^2_{L^2(\mathbb{R}_+)}+\|(V+|\gamma|)^{1/2}\varphi\|^2_{L^2(\mathbb{R}_+)}\ ,
	\end{equation*}
	setting $\gamma:=\inf V(x)$.
\end{lemma}
\begin{proof} Due to positivity of $V+|\gamma|$ one has 
	\begin{equation*}\begin{split}
	\|(-\Delta)^{1/2}\varphi_t\|^2_{L^2(\mathbb{R}_+)} &\leq \|(-\Delta+V+|\gamma|)^{1/2}\varphi_t\|^2_{L^2(\mathbb{R}_+)} \\
	&= \|(-\Delta+V+|\gamma|)^{1/2}\varphi\|^2_{L^2(\mathbb{R}_+)}\\
	&=\|(-\Delta)^{1/2}\varphi\|^2_{L^2(\mathbb{R}_+)}+\|(V+|\gamma|)^{1/2}\varphi\|^2_{L^2(\mathbb{R}_+)}\\
	&\leq \|\varphi\|^2_{1}\ .
	\end{split}
	\end{equation*}
Now the statement readily follows from [eq.(4),\cite{RadinSimon}] which states that
\begin{equation*}
\langle \varphi_t, x^2 \varphi_t\rangle_{L^2(\mathbb{R}_+)}^{1/2} \leq \langle \varphi, x^2 \varphi\rangle_{L^2(\mathbb{R}_+)}^{1/2}+2\int_{0}^{t}\|(-\Delta)^{1/2}\varphi_s\|_{L^2(\mathbb{R}_+)}\ \ud s\ .
\end{equation*}
\end{proof}

In fact, one can directly show the following. 
\begin{theorem} Define the subspace 
	\begin{equation*}
	S:=\{\varphi \in \cD_q:\ \|\varphi\|_{1}< \infty  \}\ .
	\end{equation*}
	Then $e^{-iHt}$ maps $S$ onto $S$ as a bounded operator, i.e., 
	\begin{equation}\label{EquationEstimate}
	\|e^{-iHt}\varphi\|_{1} \leq (c+d\cdot t)\cdot \|\varphi\|_{1}\ , \qquad \forall \varphi \in S\ ,
	\end{equation}
	for some constants $c,d > 0$.
\end{theorem}
\begin{proof} Due to Lemma~\ref{TimeEstimateI} and its proof, estimate~\eqref{EquationEstimate} follows from 
	\begin{equation*}\begin{split}
	\|(V+|\gamma|)^{1/2}\varphi_t\|^2_{L^2(\mathbb{R}_+)} &\leq \|(-\Delta+V+|\gamma|)^{1/2}\varphi_t\|^2_{L^2(\mathbb{R}_+)}\\
	&=\|(-\Delta+V+|\gamma|)^{1/2}\varphi\|^2_{L^2(\mathbb{R}_+)}\\ 
	&=\|(-\Delta)^{1/2}\varphi\|^2_{L^2(\mathbb{R}_+)}+\|(V+|\gamma|)^{1/2}\varphi\|^2_{L^2(\mathbb{R}_+)}\ .
	\end{split}
	\end{equation*}
	Hence $e^{-iHt}$ is a bounded operator into $S$. Now suppose that $e^{-iHt}$ is not onto $S$. Then there exists an element $\psi \in S$ such that $e^{-iHt}\varphi-\psi\neq 0$ for all $\varphi \in S$. However, choosing $\varphi:=e^{iHt}\psi \in S$ on arrives at a contradiction and the statement is proved. 
\end{proof}